\documentclass[journal]{IEEEtran}

\usepackage{amsmath,amssymb,amsthm}
\usepackage{graphicx}
\usepackage{subfigure}
\usepackage{color}
\usepackage{algorithm}
\usepackage[noend]{algpseudocode}
\usepackage{graphics}
\usepackage{verbatim,bm}
\usepackage{setspace}
\usepackage{citesort}
\usepackage{enumitem}

\newtheorem{theorem}{Theorem}

\newtheorem{corollary}{Corollary}
\newtheorem{remark}{Remark}

\begin{document}

\title{Taming Tail Latency for Erasure-coded, Distributed Storage Systems}

\author{Vaneet Aggarwal, Abubakr O. Al-Abbasi, Jingxian Fan, and Tian Lan \thanks{V. Aggarwal. A. O. Al-Abbasi, and J. Fan are with Purdue University, West Lafayette IN 47907, email: \{vaneet,aalabbas,fan137\}@purdue.edu. T. Lan is with the George Washington University, Washington, DC 20052, email: tlan@gwu.edu.  This work was presented in part at the IEEE International Conference on Computer Communications (Infocom) 2017 \cite{Jingxian}.}   }
%	\IEEEauthorblockN{Vaneet Aggarwal}
%	\IEEEauthorblockA{Purdue University\\
%		Email: vaneet@purdue.edu}
%		\and 
%	\IEEEauthorblockN{Abubakr O. Al-Abbasi}
%	\IEEEauthorblockA{Purdue University\\
%		Email: aalabbas@purdue.edu}
%	\and 
%	\IEEEauthorblockN{Jingxian Fan}
%	\IEEEauthorblockA{Purdue University\\
%		Email: fan137@purdue.edu}
%	\and
%	\IEEEauthorblockN{Tian Lan}
%	\IEEEauthorblockA{George Washington University\\
%		Email: tlan@gwu.edu}
%}

%\author{Vaneet Aggarwal, Jingxian Fan, and Tian Lan} % <-this % stops a space
%}% <-this % stops a space
\maketitle

\begin{abstract}
Distributed storage systems are known to be susceptible to long tails in response time. In modern online storage systems such as Bing, Facebook, and Amazon, the long tails of the service latency are of particular concern. with 99.9th percentile response times being orders of magnitude worse than the mean. As erasure codes emerge as a popular technique to achieve high data reliability in distributed storage while attaining space efficiency, taming tail latency still remains an open problem due to the lack of mathematical models for analyzing such systems. To this end, we propose a framework for quantifying and optimizing tail latency in erasure-coded storage systems. In particular, we derive upper bounds on tail latency in closed-form for arbitrary service time distribution and heterogeneous files. Based on the model, we formulate an optimization problem to jointly minimize weighted latency tail probability of all files over the placement of files on the servers, and the choice of servers to access the requested files. The non-convex problem is solved using an efficient, alternating optimization algorithm. Numerical results show significant reduction of tail latency for erasure-coded storage systems with realistic workload.
\end{abstract}

%	\fontsize{10.8}{13}\selectfont

\section{Introduction}

Due to emerging applications such as big data analytics and cloud computing, distributed storage systems today often store multiple petabytes of data \cite{Sathiamoorthy13,Asure14,Fikes10}. As a result, these systems are transitioning from full data replication to the use of erasure code for encoding and spreading data chunks across multiple machines and racks, in order to achieve more efficient use of storage space while maintaining high reliability despite system failures. It is shown that using erasure codes can reduce the cost of storage by more than 50\% \cite{Asure14} due to smaller storage space and datacenter footprint.

A key tradeoff for using erasure codes is performance. Distributed storage systems that employ erasure codes are known to be susceptible to long latency tails. Under full data replication, if a file is replicated $n$ times, it can be recovered from any of the $n$ replica copies. However, for an erasure-coded storage system using an $(n,k)$ code, a file is encoded into $n$ equal-size data chunks, allowing reconstruction from any subset of $k<n$ chunks. Thus, reconstructing the file requires fetching $k$ distinct chunks from different servers, which leads to significant increase of tail latency, since service latency in such systems is determined by the {\em hottest} storage nodes with highest congestion and slowest speed, which effectively become performance bottlenecks. It has been shown that in modern Web applications such as Bing, Facebook, and Amazon's retail platform, the long tail of latency is of particular concern, with $99.9$th percentile response times that are orders of magnitude worse than the mean \cite{T1,T2}. Despite mechanisms such as load-balancing and resource management, evaluations on large scale storage systems indicate that there is a high degree of randomness in delay performance \cite{Liang:2014}. The overall response time in erasure coded data-storage systems is dominated by the long tail distribution of the required parallel operations \cite{Barroso:2011}.

To the best of our knowledge, quantifying the impact of erasure coding on tail latency is an open problem for  distributed storage systems. Although  recent research progress has been made on providing bounds of mean service latency \cite{ISIT:12,Joshi:13,lee2017mds,Tian_latency,Yu_TON}, much less is known on tail latency (i.e., $x$th-percentile latency for arbitrary $x\in[0,1]$) in erasure-coded storage systems.  Mean Service latency for replication-based systems for identical servers with independent exponential service-times has been characterized for homogeneous files in \cite{gardner2016understanding}. However, the problem for erasure-coded based systems is still an open problem.  To provide an upper bound on mean service latency of homogeneous files, {\em Fork-join queue analysis} in \cite{Makowski:89,JK13,LK13,Joshi:13,CS14,KumarTC14} provides upper bounds for mean service latency by forking each file request to all storage nodes. In a separate line of work, {\em Queuing-theoretic analysis} in \cite{ISIT:12,lee2017mds} proposes a {\em block-$t$-scheduling} policy that only allows the first $t$ requests at the head of the buffer to move forward. However, both approaches fall short of quantifying tail latency due to a {\em state explosion} problem, because states of the corresponding queuing model must encapsulate not only a snapshot of the current system including chunk placement and queued requests but also past history of how chunk requests have been processed by individual nodes. Later, mean latency bounds for arbitrary service time distribution and heterogeneous files are provided in \cite{Tian_latency,Yu_TON} using order statistic analysis and a probabilistic request scheduling policy. The authors in \cite{li2016mean} used probabilistic
scheduling with uniform probabilities and exponential service
times to show improved latency performance of erasure coding
as compared to replication in the limit of large number of
servers for replication-based systems. While reducing mean latency is found to have a positive impact on pushing down the latency envelop (e.g., reducing the 90th, and 99th percentiles) \cite{Liang:2014}, quantifying and optimizing tail latency for erasure-coded storage is still an open problem.

In this paper, we propose an analytical framework to quantify tail latency in distributed storage systems that employ erasure codes to store files. This problem is challenging because (i) tail latency is significantly skewed by performance of the slowest storage nodes; (ii) a joint chunk scheduling problem needs to be solved on the fly to decide $n$-choose-$k$ chunks/servers serving each file request; and (iii) the problem is further complicated by the dependency and interference of chunk access times of different files on shared storage servers. Toward this end, we make use of probabilistic scheduling proposed in \cite{Tian_latency,Yu_TON,Yu-CCGRID,Yu-ICDCS,Yu-TCC16,Sprout}. Upon the arrival of each file request, we randomly dispatch a batch of $k$ chunk requests to $k$-out-of-$n$ storage nodes selected with some predetermined probabilities. Then, each storage node manages its local queue independently and continues processing requests in order. A file request is completed if all its chunk requests exit the system. This probabilistic scheduling policy allows us to analyze the (marginal) queuing delay distribution of each storage node and then combine the results (through Laplace Stieltjes Transform and order statistic bounds) to obtain an upper bound on tail latency in closed-form for general service time distributions. The tightest bound is obtained via an optimization over all probabilistic schedulers and all Markov bounds on tail probability.

	The proposed framework provides a mathematical crystallization of tail latency, illuminating key system design tradeoffs in erasure-coded storage.  Prior evaluation of practical systems show that the latency spread is significant even when data object sizes are in the order of megabytes \cite{Liang:2014}. To tame tail latency in erasure coded storage, we propose an optimization problem to jointly minimize the sum probability that service latency of each file exceeds a given threshold. This optimization is carried out over three dimensions: the joint placement of all files, all probabilistic schedulers, and the auxiliary variables in the tail latency bounds. We note that the probabilistic scheduler helps decrease the differentiated tail latency of the files as compared to accessing the lowest-queue servers which is important for overall tail latency of files.  Since data chunk transfer time in practical systems follows a shifted exponential distribution \cite{Yu_TON,CS14,Yu-TON16}, we show that under this assumption, the tail latency optimization can be formulated in closed-form as a non-convex minimization. To solve the problem, we prove that it is convex in two of the optimization variables and propose an alternating optimization algorithm, while the optimization with respect to file placement can be solved optimally using bipartite matching.% with a sublinear rate of convergence. While the optimization is non-covex, the proposed algorithm is proven to obtain an optimal solution to the problem using properties in \cite{1137/13094829X}. 
%Extensive simulations shows significant reduction of tail latency for erasure-coded storage systems.		
Extensive simulations shows significant reduction of tail latency for erasure-coded storage systems using the proposed optimization over five different baseline strategies.

The main contributions of this paper are summarized as follows:
\begin{itemize}[leftmargin=*]
\item We propose an analytical framework to quantify tail latency for arbitrary erasure-coded storage systems and service time distributions. 

\item When chunk transfer time follows shifted-exponential distribution, we formulate a weighted latency tail probability optimization that simultaneous minimizes tail latency of all files by optimizing the system over three dimensions: chunk placement, auxiliary variables, and the scheduling policy.

\item We develop an alternating optimization algorithm which is shown to converge to a local optima for the tail latency optimization. Two of the subproblems are convex, while bipartite matching is used to solve the third subproblem. %that is proven to be optimal for the proposed non-convex problem of tail latency optimization. 
Significant tail latency reduction up to a few orders of magnitude is validated through numerical results.

\end{itemize}

The rest of the paper is organized as follows. Section II gives the system model for the problem. Section III finds an upper bound on tail latency through probabilistic scheduling and Laplace Stieltjes transform of the waiting time from each server. Section IV formulates and solves the tail latency optimization. Section V presents our numerical results and Section VI concludes the paper. 
\section{System Model}

We consider a data center consisting of $m$ heterogeneous servers, denoted by $M={1,2,...,m}$, also called storage nodes. To distributively store a set of $r$ files, indexed by $i=1,2,...r$, we partition each file $i$ into $k_i$ fixed-size chunks and then encode it using an $(n_i, k_i)$ MDS erasure code to generate $n_i$ distinct chunks of the same size for file $i$. The encoded chunks are assigned to and stored on $n_i$ distinct storage nodes, represented by a set $S_i$ of storage nodes, satisfying $S_i \subseteq M$ and $n_i=\vert S_i \vert$.  The use of $(n_i, k_i)$ MDS erasure code allows the file to be reconstructed from any subset of $k_i$-out-of-$n_i$ chunks, whereas it also introduces a redundancy factor of $n_i/k_i$. Thus, upon the arrival of each file request, $k_i$ distinct chunks are selected by a scheduler and retrieved to reconstruct the desired file. Figure~\ref{sys} illustrates a distributed storage system with 7 nodes. Three files are stored in the system using $(6,4)$, $(5,3)$, and $(3,2)$ erasure codes, respectively. File requests arriving at the system are jointly scheduled to access $k_i$-out-of-$n_i$ distinct chunks. Prior work analyzing erasure-coded storage systems mainly focus on mean latency, including two approaches using queuing-theoretic analysis in \cite{ISIT:12,lee2017mds} and fork-join queue analysis in \cite{Makowski:89,JK13,LK13,Joshi:13,CS14,KumarTC14}.

\begin{figure*}
	\begin{center}
		\includegraphics[width=0.75\textwidth]{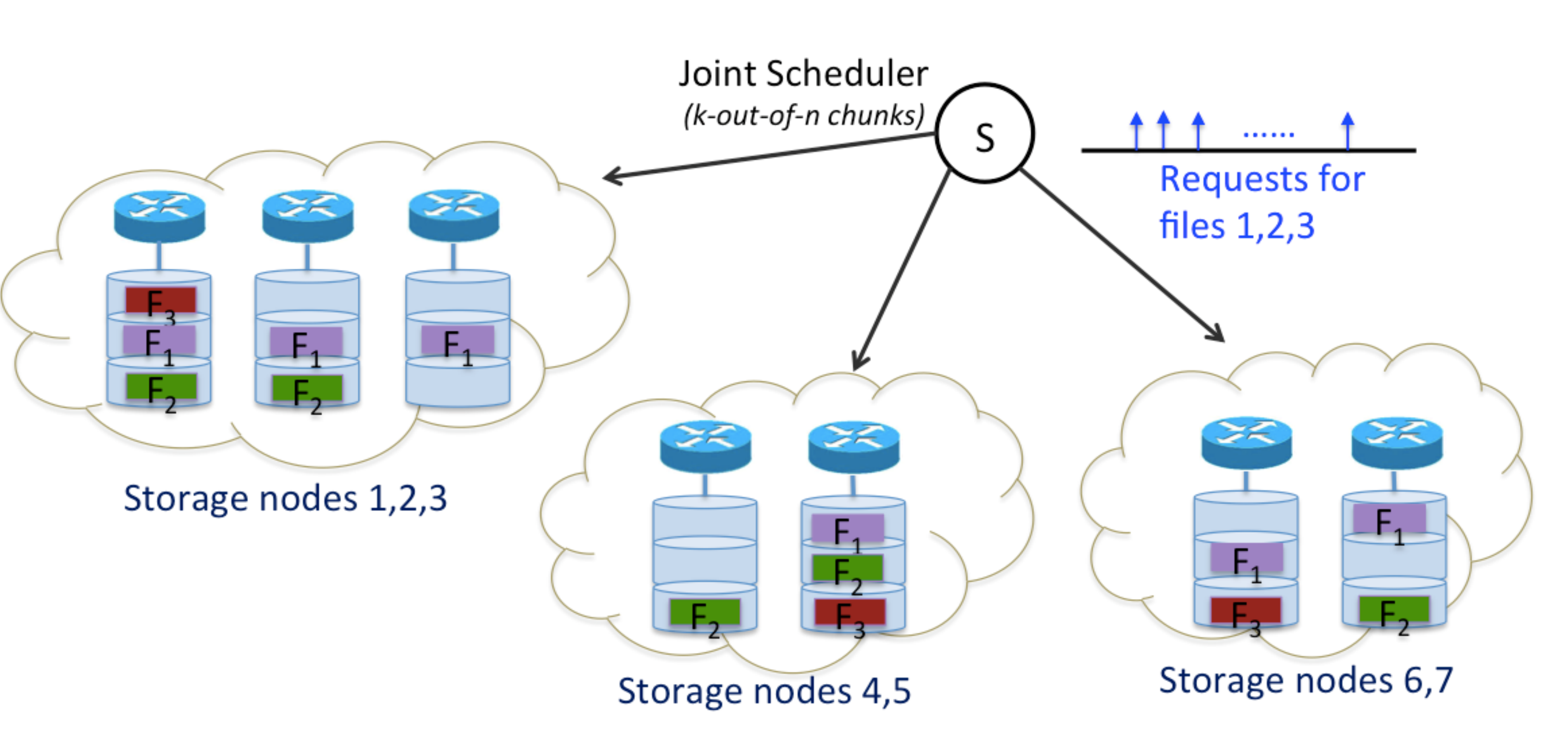}
		\caption{An illustration of a distributed storage system equipped with $7$ nodes and storing $3$ files using different erasure codes.}
		\label{sys}
	\end{center}
\end{figure*}

However, both approaches fall short of quantifying tail latency, because states of the corresponding queuing model must encapsulate not only a snapshot of the current system including chunk placement and queued requests, but also past history of how chunk requests have been processed by individual nodes. This leads to a {\em state explosion} problem as practical storage systems usually handle a large number of files and nodes \cite{Yu_TON}. To the best of our knowledge, quantifying tail latency for erasure-coded storage system is still an open problem because of challenges in joint request scheduling (i.e., selecting $n$-choose-$k$ chunks for each request on the fly with the goal of minimizing tail latency) as well as the dependency of straggling fragments on hot storage nodes. Consider the erasure-coded storage system storing 3 files, as shown in Figure~\ref{sys}. It is easy to see that a simple scheduling policy that accesses available chunks with equal probability lead to high tail latency, which is determined
by hot storage nodes (i.e., nodes 1 and 5 in this case) with slowest performance. Yet, a policy that load-balances the number of requests processed by each server does not necessarily optimize tail latency of all files, which employ different erasure codes resulting in different impact on service latency. Assuming that chunk transfer time from all storage nodes have the same distribution, file 1 using $(6,4)$ code could still have much higher tail latency than file 3 that uses $(3,2)$ code, since its service time of each file request is determined by the slowest of the 4 selected chunks (rather than 2 selected chunks).

In this paper, we use the Probabilistic Scheduling from \cite{Yu_TON,Tian_latency}, which is a probabilistic scheduling policy: 1) dispatches each batch of chunk requests (corresponding to the same file request) to appropriate a set of nodes (denoted by set $A_i$ of servers for file $i$) with predetermined probabilities ($P(A_i) $ for set $A_i$ and file $i$); 2) each node buffers requests in a local queue and processes in order. The authors of \cite{Yu_TON,Tian_latency} have shown that a probabilistic scheduling policy with feasible probabilities $\{P(A_i): \forall i,A_i\}$ exists if and only if there exists conditional probabilities $\pi_{i,j} \in [0,1],\forall i,j$ satisfying $$
\sum_{j=1}^m\pi_{i,j}=k_i \quad \forall i \quad \text{ and } \quad \pi_{i,j}=0\ \text{ if } \ j \notin S_i. $$ Consider the example shown in Figure~\ref{sys}. Under probabilistic scheduling, upon the arrival of a file 1 request, we randomly select $k_1=4$ nodes (from $\{1,2,3,5,6,7\}$) with available file chunks with respect to known probabilities $\{\pi_{1,j}, \ \forall j\}$ and dispatch a chunk request to each selected storage node. 
Then, each storage node manages its local queue independently and continues processing requests in order. The file request is completed if all its chunk requests are processed by individual nodes. While this probabilistic scheduling is used to provide an upper bound on mean service time in \cite{Yu_TON,Tian_latency}, we extend the policy and provide an analytical model for tail latency, enabling a novel tail latency optimization.

We will now describe a queueing model of the distributed storage system. We assume that the arrival of client requests for each file $i$ form an independent Poisson process with a known rate $\lambda_i$. We consider chunk service time $X_j$ of node $j$ with arbitrary distributions, whose statistics can be obtained inferred from existing work on network delay \cite{WK,AY11} and file-size distribution \cite{Corbett:2004,Calder:2011}. Under MDS codes, each file $i$ can be retrieved from any $k_i$ distinct nodes that store the file chunks. We model this by treating each file request as a batch of $k_i$ chunk requests, so that a file request is served when all $k_i$ chunk requests in the batch are processed by distinct storage nodes. Even though the choice of codes for different files can be different, we assume that the chunk size is the same for all files.  All requests are buffered in a common queue of infinite capacity.

%The latency for file $i$ is noted as $L_i$, and file $i$ has a weight of itself as $\omega_i$, with this model we use all the variables above to minimize the probability of mean latency more than a constant $x$.

\begin{table}[ht]
\centering
\caption{Main notations Used in this paper}
\label{my-label}
\resizebox{.48\textwidth}{!}{  
	\begin{tabular}{cc}
\hline \hline
Symbol        & Meaning                                               \\  \hline
$r$             & Number of files in system by $i=1,2,...,r $             \\
$m$             & Number of storage  nodes                        \\
($n_i, k_i$)        & Erasure code parameters for file $i$                        \\
$\lambda_i$    & Arrival rate of file $i$                                \\
$\pi_{ij}$   & Probability of retrieving chunk of file $i$ from node $j$  \\
$L_i$          & Latency of retrieving file $i$                          \\
$x$             & Parameter indexing latency tail probability                              \\
$Q_j$          & Sojourn Time of node $j$                                \\
$X_j$          & Chunk Service Time of node $j$                                \\
$M_j(t)$          & Moment Generating Function for the service time of node $j$                                \\
$\mu_j$        & Mean service time of node $j$                           \\
$\Lambda_j$    & Arrival rate on node $j$                                \\
$\rho_j$       & Request intensity at node $j$                           \\
$S_i$	&Set of storage nodes having chunks from file $i$\\
$A_i $         & Set of nodes used to provide chunks from file $i$                \\
$(\alpha_j,\beta_j)$ & Parameters of Shifted Exponential distributed\\& service time at node $j$\\
$\omega_i$     & weight of file $i$                                      \\ \hline \hline
\end{tabular}}
\end{table}
 \section{Bounds on Tail Latency}

We first quantify tail latency for erasure-coded storage systems with arbitrary service time distribution (i.e., arbitrary known distribution of $X_j$). Let ${\bf Q}_j$ be the (random) time the chunk request spends in node $j$ (sojourn time). Under probabilistic scheduling, the service time (denoted by $L_i$) of a file-$i$ request is determined by the maximum chunk service time at a randomly selected set $A_i$ of storage nodes.

Under probabilistic scheduling, the arrival of chunk requests at node $j$ form a Poisson Process with rate $\Lambda_j = \sum_i \lambda_i \pi_{ij}$. Let $M_j(t) = {\mathbb E}[e^{t X_j}]$ be the moment generating function of service time of processing a single chunk at server $j$. Then, the Laplace Stieltjes Transform of ${\bf Q}_j$  is given, using Pollaczek-Khinchine formula, as 
\begin{equation}
{\mathbb E}[e^{-s Q_j}] = \frac{(1-\rho_j)s M_j(-s)}{s-\Lambda_j(1-M_j(-s))},\label{polla}
\end{equation}
where $\rho_j  =  \Lambda_j{\mathbb E}[X_j]$ is the request intensity at node $j$, and $M_j(t) = {\mathbb E}[e^{t X_j}]$ is the moment generating function of $X_j$ \cite{kleinrock1976queueing}. Further, let the latency of the file $i$ be denoted as $L_i$ using probabilistic scheduling. The {\em latency tail probability} of file $i$ is defined as the probability that $L_i$ is greater than or equal to $x$, for a given $x$. For given weight $w_i$ for file $i$, this paper wishes to minimize $\sum_i w_i \Pr(L_i\ge x)$. Since finding $\Pr(L_i\ge x)$ in closed form is hard for general service time distribution, we further use an upper bound on this and use that instead of $\Pr(L_i\ge x)$ in the objective. 

The following theorem gives an upper bound on the latency tail probability of a file.

 \begin{theorem}
 The latency tail probability for file $i$, $\Pr(L_i\ge x)$ using probabilistic scheduling is bounded by
 \begin{equation}
  \Pr(L_i\ge x)\le \sum_{j}\frac{\pi_{ij}}{e^{t_j x}} \frac{(1-\rho_j)t_j M_j(t_j)}{t_j - \Lambda_j (M_j(t_j)-1)},
\end{equation}
for any $t_j> 0$, $\rho_j = \Lambda_j{\mathbb E}[X_j]$, satisfying $M_j(t_j)<\infty$ and $\Lambda_j (M_j(t_j)-1)<t_j$.
 \end{theorem}

\begin{proof}
We consider an upper bound on latency tail probability using probabilistic scheduling as follows.
 \begin{eqnarray}
 \Pr(L_i\ge x) &\stackrel{\text{(a)}}{=}& \Pr_{A_i, Q_j} (\max_{j\in A_i} Q_j \ge x) \\
  &=& \Pr_{A_i, Q_j} (Q_j \ge x \text{ for some } j \in A_i) \\
  &=& {\mathbb E}_{A_i,Q_j} [\max_{j\in A_i} {1}_{(Q_j\ge x)}]\\
  &\le&  {\mathbb E}_{A_i,Q_j}\sum_{j\in A_i} [ {1 }_{(Q_j\ge x)}]\\
  &=&  {\mathbb E}_{A_i}\sum_{j\in A_i}[ \Pr{(Q_j\ge x)}]\\
  &=& \sum_{j}\pi_{ij} [ \Pr{(Q_j\ge x)}],
 \end{eqnarray}
where (a) follows since for probabilistic scheduling, the time to retrieve the file is the maximum of the time of retrieving all the chunks from $A_i$.

Using Markov Lemma, we have $ \Pr{(Q_j\ge x)} \le \frac{{\mathbb E}[e^{t_j Q_j}]}{e^{t_j x}}$. In order to obtain ${\mathbb E}[e^{t_j Q_j}]$, we use Pollaczek-Khinchine formula for Laplace Stieltjes Transform of $Q_j$ in \eqref{polla} and use $s=-t_j$. However, the expression is finite only when $\Lambda_j (M_j(t_j)-1)<t_j$. This proves the result as in the statement of the Theorem.
\end{proof}

In some cases, the moment generating function may not exist, which means that the condition $\Lambda_j (M_j(t_j)-1)<t_j$ may not be satisfied for any $t_j>0$. In such cases, we will use the Laplace Stieltjes Transform directly to give another upper bound in the next theorem.

\begin{theorem}
 The latency tail probability for file $i$, $\Pr(L_i\ge x)$ is bounded by
 \begin{equation}
  \Pr(L_i\ge x)\le \sum_{j}\frac{\pi_{ij}(1 - {\mathbb E}[e^{-s_j Q_j}])}{1  - e^{-s_j x}},
\end{equation}
for any $s_j> 0$, where $\rho_j = \Lambda_j{\mathbb E}[X_j]$, ${\mathbb E}[e^{-s Q_j}] = \frac{(1-\rho_j)s L_j(s)}{s-\Lambda_j(1-L_j(s))}$, and $L_j(s) = E[e^{-s X_j}]$.
 \end{theorem}
 \begin{proof}
 This result is a variant of Theorem 1, where Markov Lemma is used using Laplace Stieljes Transform of the Queue Waiting Time rather than the moment generating function.
 \end{proof}

We next consider the case when the service time distribution is a shifted exponential distribution. This choice is motivated by the Tahoe experiments \cite{Yu_TON} and Amazon S3 experiments \cite{CS14}.  Let the service time distribution from server $j$ has probability density function $f_{X_j}(x)$, given as
 \begin{equation}
   f_{X_j}(x)= \left\{\begin{array}{lr}
   \alpha_j e^{-\alpha_j(x-\beta_j)}, & \text{ for }  x\ge \beta_j\\
   0,  & \text{ for }  x< \beta_j
   \end{array}.\right.
 \end{equation}
 Exponential distribution is a special case with $\beta_j=0$. The Moment Generating Function is given as
 \begin{equation}
 M_j(t) = \frac{\alpha_j}{\alpha_j-t}e^{\beta_j t} \quad \text{ for } t< \alpha_j.
 \end{equation}
 Using these expressions, we have the following result.

 \begin{corollary}
  When the service time distributions of servers are given by shifted exponential distribution, the latency tail probability for file $i$, $\Pr(L_i\ge x)$, is bounded by
  \begin{equation}
  \Pr(L_i\ge x)\le \sum_{j}\frac{\pi_{ij}}{e^{t_j x}} \frac{(1-\rho_j)t_j M_j(t_j)}{t_j - \Lambda_j (M_j(t_j)-1)},
\end{equation}
for any $t_j>0$, $\rho_j = \frac{\Lambda_j}{\alpha_j}+\Lambda_j \beta_j$, $\rho_j<1$,  and $t_j(t_j - \alpha_j + \Lambda_j) +\Lambda_j \alpha_j (e^{\beta_j t_j}-1)<0$.
 \end{corollary}
 \begin{proof}
 We note that the condition $\Lambda_j (M_j(t_j)-1)<t_j$ reduces to $t_j(t_j - \alpha_j + \Lambda_j) +\Lambda_j \alpha_j (e^{\beta_j t_j}-1)<0$. Since $t_j\ge \alpha_j$ will not satisfy $t_j(t_j - \alpha_j + \Lambda_j) +\Lambda_j \alpha_j (e^{\beta_j t_j}-1)<0$, the conditions in the statement of the Corollary implies $t_j<\alpha_j$ where the above moment generating function expression is used. 
 \end{proof}

Since exponential distribution is a special case of the shifted exponential distribution, we have the following corollary. 

 \begin{corollary}
  When the service time distributions of servers are given by exponential distribution, the latency tail probability for file $i$, $\Pr(L_i\ge x)$, is bounded by
  \begin{equation}
  \Pr(L_i\ge x)\le \sum_{j}\frac{\pi_{ij}}{e^{t_j x}} \frac{(1-\rho_j)t_j M_j(t_j)}{t_j - \Lambda_j (M_j(t_j)-1)},
\end{equation}
for any $t_j>0$, $\rho_j = \frac{\Lambda_j}{\alpha_j}$, $\rho_j<1$,  $t_j <\alpha_j(1-\rho_j)$
 \end{corollary}

 \section{Optimizing Weighted  Latency Tail Probability}

Now we formulate a joint latency tail probability optimization for multiple, heterogeneous files. Since the latency tail probability  is given by $\Pr(L_i \ge x)$ for $x> \max_j \beta_j$, we consider an optimization that minimizes {\em weighted latency tail probability} of all files, defined by
\begin{eqnarray}
\sum_i \omega_i\Pr(L_i \ge x),
\end{eqnarray}
where $\omega_i = \frac{\lambda_i}{\sum_i \lambda_i}$ is a positive weight assigned to file $i$ so that the files with larger arrival rates are weighted higher, and latency tail probability of file-$i$ service time is $\Pr(L_i \ge x)$. We consider the proposed bound on the latency tail probability to have the objective function as 
\begin{eqnarray}
 \sum_i \lambda_i\left (\sum_{j}\frac{\pi_{ij}}{e^{t_j x}} \frac{(1-\rho_j)t_j M_j(t_j)}{t_j - \Lambda_j (M_j(t_j)-1)} \right). 
\end{eqnarray}
%
%
%By adjusting weights $\omega_i$, the proposed optimization allows us to explore a tail-latency tradeoff between different files and to offer elastic Service Level Agreements (SLA) to users with different tail latency preference. For instance, a large weight $\omega_i$ can be assigned to a video streaming application requiring quick responses, while a small $\omega_i$ might be appropriate for online data backup that is latency insensitive.
%

Let  $\bm \pi = \{\pi_{i,j} \forall i,j\}$, ${\mathbf t}=\{t_j \forall j\}$, and ${\bm{  \mathcal{S} }} = \{\mathcal{S}_{i} \forall i\}$. We consider the following Weighted  Latency Tail Probability (WLTP) optimization problem over the scheduling probabilities $\bm \pi$, the placement of files ${\bm{  \mathcal{S} }}$,  and auxiliary parameters ${\mathbf t}$, i.e.,
\begin{eqnarray}
&\min&  \sum_{j}\Lambda_j{e^{-t_j x}} \frac{(1-\rho_j)t_j M_j(t_j)}{t_j - \Lambda_j (M_j(t_j)-1)}  \label{objeq} \\
&s.t.&\Lambda_j=\sum_i \lambda_{i}\pi_{ij}\label{formf}\\
&& M_j(t) = \frac{\alpha_j}{\alpha_j-t}e^{\beta_j t}\label{mjt}\\
&& \rho_j = \frac{\Lambda_j}{\alpha_j}+\Lambda_j \beta_j \label{opt_1}\\
&&\sum_j\pi_{i,j} =k_i \label{opt_3}\\
&& \pi_{i,j}=0, j\notin \mathcal{S}_{i} \label{rhocon}\\
&&\pi_{i,j}\in [0,1] \label{opt_5}\\
&&\left|\mathcal{S}_{i}\right|=n_{i}\,,\,\,\left|\mathcal{A}_{i}\right|=k_{i}\,\,\,\text{\ensuremath{\forall}i} \label{plc}\\
%&& \rho_j < 1 \label{opt_2} \\
&& t_j\ge 0\label{tjo}\\
&& \!\!\!\! t_j(t_j - \alpha_j + \Lambda_j) +\Lambda_j \alpha_j (e^{\beta_j t_j}-1)<0 \label{forml}\\
&var.& \bm \pi, {\mathbf t}, {\bm{  \mathcal{S} }}
\end{eqnarray}

Here, Constraint (\ref{formf}) gives the aggregate arrival rate $\Lambda_j$ for each node under give scheduling probabilities $\pi_{i,j}$ and arrival rates $\lambda_i$, Constraint (\ref{mjt}) defines moment generating function with respect to parameter $t_j$, Constraint (\ref{opt_1}) defines the traffic intensity of the servers, Constraints (\ref{opt_3}-\ref{opt_5}) guarantee that the scheduling probabilities are feasible, and finally, the moment generating function exists due to the technical constraint in (\ref{forml}). If  (\ref{forml}) is satisfied, $\rho_j<1$ holds too thus ensuring the stability of the storage system (i.e., queue length does not blow up to infinity under given arrival rates and scheduling probabilities). We note that $t_j>0$ can be equivalently converted  to $t_j\ge 0$ (and thus done in \eqref{tjo}) since $t_j=0$ do not satisfy $t_j(t_j - \alpha_j + \Lambda_j) +\Lambda_j \alpha_j (e^{\beta_j t_j}-1)<0$ and has already been accounted for.  We note that the the optimization over $\bm \pi$ helps decrease the weighted tail latency probability and gives significant flexibility over choosing the lowest-queue servers for accessing the files. The placement of the files ${\bm{  \mathcal{S} }}$ helps separate the highly accessed files on different servers thus reducing the objective. Finally, the optimization over the auxiliary variables  ${\mathbf t}$  gives a tighter bound on the weighted latency tail probability. 

% as compared to accessing the lowest-queue servers which is important for overall tail latency of files.

\begin{remark}
The proposed WLTP optimization is non-convex, since Constraint (\ref{forml}) is non-convex in ($\bm \pi$, ${\mathbf t}$). Further, the content placement ${\bm{  \mathcal{S} }}$ has integer constraints. 
\end{remark}

To develop an algorithmic solution, we prove that the problem is convex individually with respect to  the optimization variables  ${\mathbf t}$ and ${\bm \pi}$, when the other variables are fixed. This result allows us to propose an alternating optimization algorithm for the problem.   %, which is later proven to be indeed optimal. We first show a lemma that will be used to prove the convexity results.
The next result shows the the problem is convex  in ${\mathbf t} = (t_1, t_2, \cdots, t_m)$.

\begin{theorem} 
The objective function, $\sum_{j}\frac{\Lambda_j}{e^{t_j x}} \frac{(1-\rho_j)t_j M_j(t_j)}{t_j - \Lambda_j (M_j(t_j)-1)}$ is convex in ${\mathbf t} = (t_1, t_2, \cdots, t_m)$ in the region where the constraints in \eqref{formf}-\eqref{forml} are satisfied.
\end{theorem}
\begin{proof}
	We note that inside the summation, the term only depends on a single value of $t_j$. Thus, it is enough to show that $ \frac{t_j e^{-t_j x}M_j(t_j)}{t_j - \Lambda_j (M_j(t_j)-1)}$ is convex with respect to $t_j$. Since there is only a single index $j$ here, we ignore this subscipt for the rest of this proof.
	
	We denote
	\begin{eqnarray}
	F(t)&=&\frac{te^{-t x}M(t)}{t- \Lambda (M(t)-1)}\\
	&=&\frac{\alpha t e^{(\beta-x)t}}{-t^2+(\alpha-\Lambda)t+\Lambda\alpha-\Lambda\alpha e^{\beta t}}\\
	&=&\frac{\alpha t e^{(\beta-x)t}}{-t^2+(\alpha-\Lambda)t-\Lambda\alpha (e^{\beta t}-1)}\\
	&=&\frac{\alpha t e^{(\beta-x)t}}{-t^2+(\alpha-\Lambda)t-\Lambda\alpha \sum_{u=1}^\infty\frac{(\beta t)^u}{u!}}\\
	&=&\frac{\alpha  e^{(\beta-x)t}}{-t+(\alpha-\Lambda)-\Lambda\alpha \sum_{u=1}^\infty\frac{(\beta)^u t^{u-1}}{u!}}
	\end{eqnarray}

Thus, $F(t)$ can be written as product of $f(t) = \alpha  e^{(\beta-x)t}$ and $g(t)= \frac{1}{h(t)}$, where $h(t) = -t+(\alpha-\Lambda)-\Lambda\alpha \sum_{u=1}^\infty\frac{(\beta)^u t^{u-1}}{u!}$. Since the constraints in \eqref{formf}-\eqref{forml} are satisfied, $h(t)> 0$. Further, all positive deriavatives of $h(t)$ are non-positive. Let $w(t) = -h'(t)$. Then, $w(t)\ge 0$, and $w'(t)\ge 0$.

Further, we have
	\begin{eqnarray}
	g(t)&=&\frac{1}{h(t)}\nonumber\\
	g'(t)&=&\frac{w(t)}{h^2(t)}\nonumber\\
	g''(t)&=&\frac{h(t)w'(t)+2w^2(t)}{h^3(t)}.
	\end{eqnarray}

	\fontsize{10.4}{12.5}\selectfont
Using these, $	F''(t)$ is given as 
\begin{eqnarray}
	&&F''(t)\nonumber\\
	&=&f''(t)g(t)+f(t)g''(t)+2f'(t)g'(t)\nonumber\\
	&=&\alpha e^{(\beta-x)t}\left(((\beta-x)^2 g(t)+g''(t)+2(\beta-x)g'(t))\right)\nonumber\\
	&=&\frac{\alpha e^{(\beta-x)t}}{h^3(t)}\left((\beta-x)^2 h^2(t)+h(t)w'(t)+2w^2(t)\right.\nonumber\\
	&&\left. +2(\beta-x)w(t)h(t)\right)\nonumber\\
	&=&\frac{\alpha e^{(\beta-x)t}}{h^3(t)} \left(2\left(\frac{(\beta-x)h(t)}{2}+w(t)\right)^2
	+h(t)w'(t) \right.\nonumber\\&&\left.+\frac{(\beta-x)^2h^2(t)}{4} \right)\nonumber\\
	&\ge& 0,
	\end{eqnarray}
	where the last step follows since $h(t)\ge 0$, and $w'(t)\ge 0$. Thus, the objective function is convex in ${\mathbf t} = (t_1, t_2, \cdots, t_m)$.
\end{proof}

	\fontsize{10.8}{13}\selectfont

The next result shows that the proposed problem is convex in  ${\bm \pi} = (\pi_{ij} \forall i=1, \cdots, r, j = 1,\cdots, m)$.

\begin{theorem}
The objective function, $\sum_{j}\frac{\Lambda_j}{e^{t_j x}} \frac{(1-\rho_j)t_j M_j(t_j)}{t_j - \Lambda_j (M_j(t_j)-1)}$  is convex in ${\bm \pi} = (\pi_{ij} \forall (i,j))$.
\end{theorem}
\begin{proof}
Since the sum of convex functions is convex, it is enough to show that $\Lambda_j \frac{(1-\rho_j)}{t_j - \Lambda_j (M_j(t_j)-1)}$ is convex. Since $\Lambda_j$ is a linear function of  ${\bm \pi}$, it is enough to prove that $\Lambda_j \frac{(1-\rho_j)}{t_j - \Lambda_j (M_j(t_j)-1)}$ is convex in $\Lambda_j $. Let $H_j = \frac{1-\rho_j }{1 - \Lambda_j (M_j(t_j)-1)/t_j}$. We need to show that $\Lambda_j H_j$ is convex in $\Lambda_j $. 

We will first show that $H_j$ is increasing and convex in $\Lambda_j $.  
%
%\begin{equation}
%F_{i,j} = \pi_{ij} H_j,
%\end{equation}
%where $H_j = \frac{1-\rho_j }{1 - \Lambda_j (M_j(t_j)-1)/t_j}$ is convex w.r.t. ${\bm \pi}$. We first show that $H_j$ is convex w.r.t. ${\bm \pi}$ with non-negative gradient. In order to see that, we first note that $\Lambda_j$ is linear function of ${\bm \pi}$ with non-negative gradients. Since $H_j$ depends on ${\bm \pi}$ only through $\Lambda_j$, it is enough to show that $H_j$ is convex w.r.t. $\Lambda_j$. 
We note that $H_j$ can be written as
\begin{equation}
H_j = \frac{1-\Lambda_j C_1}{1-\Lambda_j C_2},
\end{equation}
where $C_1 = \frac{1}{\alpha_j}+\beta_j$ and $C_2 = \frac{M_j(t_j)-1}{t_j}$. Further $C_2\ge C_1$ since $M_j(t_j) -1 = {\mathbb E}[e^{t_jX_j}]-1 \ge {\mathbb E}[1+{t_jX_j}] -1 = t_j{\mathbb E}[{X_j}] = t_j\left(\frac{1}{\alpha_j}+\beta_j\right)$.
Differentiating $H_j$ w.r.t. $\Lambda_j$, we have
\begin{eqnarray}
\frac{\delta}{\delta \Lambda_j} H_j &=& \frac{C_2-C_1}{(1-\Lambda_j C_2)^2}\ge 0\\
\frac{\delta^2}{\delta \Lambda_j^2} H_j &=& 2C_2\frac{C_2-C_1}{(1-\Lambda_j C_2)^3} \ge 0.
\end{eqnarray}

Thus, $H_j$ is an increasing convex function of $\Lambda_j$. Since $\Lambda_j$ is also an increasing convex function of $\Lambda_j$ and the product of two increasing convex functions is convex, the result follows. 
%Thus, $\nabla \pi_{ij}$ is convex w.r.t. ${\bm \pi}$ with non-negative gradients. $\pi_{ij}$ is also convex w.r.t. ${\bm \pi}$, with non-negative gradients. By Lemma \ref{prodlemms}, we note that to show that the product of these is convex, it is enough to show that $\nabla \pi_{ij}  (\nabla H_j)^*$ is positive semi-definite which is true in this case since only one row is non-zero, and that row only has non-negative elements.
\end{proof}

\section{Algorithm for WLTP Optimization}

We note that the WLTP optimization problem is convex with respect to individual ${\mathbf t}$ and ${\bm \pi}$. We note that the strict $<$ constraint can be modified as $\le -\epsilon$ for an $\epsilon$ small enough. The constraints are also convex in each of the variables individually.  We will now develop an alternating minimization algorithm to solve the problem. In order to describe the Algorithm, we first define the three sub-problems:

{\bf ${\mathbf t}$-Optimization: } Input ${\bm \pi}, \mathcal{\boldsymbol{S}}$
\begin{eqnarray}
&\min&  \sum_i \omega_i\left (\sum_{j}\frac{\pi_{ij}}{e^{t_j x}} \frac{(1-\rho_j)t_j M_j(t_j)}{t_j - \Lambda_j (M_j(t_j)-1)} \right)\nonumber\\
&s.t.& \eqref{formf}, \eqref{mjt},\eqref{opt_1},\eqref{tjo}, \eqref{forml}\nonumber\\
%
%
%\Lambda_j=\sum_i \lambda_{i}\pi_{ij}\\
%&& M_j(t) = \frac{\alpha_j}{\alpha_j-t}e^{\beta_j t}\\
%&& \rho_j = \frac{\Lambda_j}{\alpha_j}+\Lambda_j \beta_j\\
%&& t_j>0\\
%&& t_j(t_j - \alpha_j + \Lambda_j) +\Lambda_j \alpha_j (e^{\beta_j t_j}-1)<0\\
&var.& {\mathbf t} \nonumber
\end{eqnarray}

{\bf ${\bm \pi}$-Optimization: } Input ${\mathbf t}, \mathcal{\boldsymbol{S}}$
\begin{eqnarray}
&\min&  \sum_i \omega_i\left (\sum_{j}\frac{\pi_{ij}}{e^{t_j x}} \frac{(1-\rho_j)t_j M_j(t_j)}{t_j - \Lambda_j (M_j(t_j)-1)} \right)\nonumber\\
&s.t.&  \eqref{formf}, \eqref{mjt},\eqref{opt_1},\eqref{opt_3}, \eqref{rhocon}, \eqref{opt_5}, \eqref{forml}\nonumber\\
%\Lambda_j=\sum_i \lambda_{i}\pi_{ij}\\
%&& M_j(t) = \frac{\alpha_j}{\alpha_j-t}e^{\beta_j t}\\
%&& \rho_j = \frac{\Lambda_j}{\alpha_j}+\Lambda_j \beta_j\\
%&&\sum_j\pi_{i,j} =K_i\\
%&&\pi_{i,j}\in [0,1]\\
%&& \pi_{i,j}=0, j\notin A_i \\
%&& t_j(t_j - \alpha_j + \Lambda_j) +\Lambda_j \alpha_j (e^{\beta_j t_j}-1)<0\\
&var.& {\bm \pi}\nonumber
\end{eqnarray}

{\bf ${\bm{  \mathcal{S} }}$-Optimization: } Input ${\mathbf t}, {\bm \pi}$
\begin{eqnarray}
&\min&  \sum_i \omega_i\left (\sum_{j}\frac{\pi_{ij}}{e^{t_j x}} \frac{(1-\rho_j)t_j M_j(t_j)}{t_j - \Lambda_j (M_j(t_j)-1)} \right)\nonumber\\
&s.t.&  \eqref{formf}, \eqref{mjt},\eqref{opt_1},\eqref{opt_3}, \eqref{rhocon}, \eqref{opt_5}, \eqref{plc}\nonumber\\
%\Lambda_j=\sum_i \lambda_{i}\pi_{ij}\\
%&& M_j(t) = \frac{\alpha_j}{\alpha_j-t}e^{\beta_j t}\\
%&& \rho_j = \frac{\Lambda_j}{\alpha_j}+\Lambda_j \beta_j\\
%&&\sum_j\pi_{i,j} =K_i\\
%&&\pi_{i,j}\in [0,1]\\
%&& \pi_{i,j}=0, j\notin A_i \\
%&& t_j(t_j - \alpha_j + \Lambda_j) +\Lambda_j \alpha_j (e^{\beta_j t_j}-1)<0\\
&var.& \bm{  \mathcal{S} } \nonumber
\end{eqnarray}

The first two sub-problems ({\bf ${\mathbf t}$-Optimization} and {\bf ${\mathbf \pi}$-Optimization}) are convex, and thus can be solved by Projected Gradient Descent Algorithm.

For the placement sub-problem ({\bf ${\mathcal{\boldsymbol{S}}}$-Optimization}), we consider optimizing over $\mathcal{S}$ for each file request separately with fixed
$\bm \pi$ and ${\mathbf t}$.  We first rewrite the latency tail probability for file $i$, $\mathbb{P}({L_i}\geq x)$ as follows

\begin{eqnarray}
\mathbb{P}({L_i}\geq x)  &{\leq}& \sum_{j}\frac{\pi_{ij}}{e^{t_j x}} \frac{(1-\rho_j)t_j M_j(t_j)}{t_j - \Lambda_j (M_j(t_j)-1)} \\
  &=& \sum_j \frac{\pi_{ij}}{e^{t_j x}} F\left(\Lambda_j\right),
 \end{eqnarray}
where $F\left(\Lambda_j\right)=\frac{(1-\Lambda_j(1/\alpha_j+\beta_j))t_j M_j(t_j)}{t_j - \Lambda_j (M_j(t_j)-1)}$. To show that the placement sub-problem can be cast
into a bipartite matching, we consider the problem of placing file
$i$ chunks on $m$ available servers. Note that placing the chunks
is equivalent to permuting the existing access probabilities $\left\{ \pi_{ij},\,\forall i\right\} $
on all $m$ servers, because $\pi_{ij}>0$ only if a chunk of file
$i$ is placed on server $j$. Let $\beta$ be a permutation of $m$
elements. Under new placement defined by $\beta$, the new probability
of accessing file $i$ chunk on server $j$ becomes $\pi_{i,\beta(j)}$.
Thus, our objective in this sub-problem is to find such a placement
(or permutation) $\beta(j)$ $\forall j$ that minimizes the average
tail probability, which can be solved via a matching problem between
the set of existing scheduling probabilities $\left\{ \pi_{ij},\,\forall i\right\} $
and the set of $m$ available servers, with respect to their load
excluding the contribution of file $i$ itself. Let $\varLambda_{j}^{-i}=\varLambda_{j}-\lambda_{i}\pi_{ij}$
when request for file $i$, be the total request rate at server $j$,
excluding the contribution of file $i$. We define a complete bipartite
graph $\mathcal{G}_{r}=\left(\mathcal{U},\mathcal{V},\mathcal{E}\right)$
with disjoint vertex sets $\mathcal{U},\mathcal{V}$ of equal size
$m$ and edge weights given by

\begin{equation}
D_{u,v}=\sum_{i}\frac{\lambda_{i}\pi_{iu} }{ e^{t_u x}} F\left(\Lambda_v^{-i}+\lambda_i \pi_{iu} \right),
\,\,\forall u,v \, ,
\end{equation}
which quantifies the contribution to overall latency tail probability by assigning existing $\pi_{iu}$ to
server $v$ that has an existing load $\varLambda_{j}^{-i}$. It can
be shown that a minimum-weight matching on $\mathcal{G}_{r}$ finds
an optimal $\beta$ to minimize

\begin{equation}
\sum_{u=1}^{m}D_{u,\beta(u)}=\sum_{u=1}^{m}\sum_{i=1}^{r}\frac{\lambda_{i}\,\pi_{iu} }{ e^{t_j x}} F\left(\Lambda_{\beta(u)}^{-i}+\lambda_i \pi_{iu}\right),
\end{equation}

The proposed algorithm
for solving latency tail probability problem is shown in Algorithm 1, where the order of files whose placements are optimized one after the other are chosen at random. Note that the order of the files whose decisions are changed make a difference. In the proposed algorithm, we take a single pass over the files since there is an outer loop of alternating minimization.

\begin{algorithm}[t]
\caption{ Proposed algorithm for solving latency tail probability Problem}

Initialize $k=0$, $\epsilon>0$,

Initialize feasible $\left\{ t_{i}(0),\pi_{i,j}(0),\right\} $

\textbf{while} $\mbox{obj}\left(k\right)-\mbox{obj}\left(k-1\right)\geq\epsilon$

$\quad$//\textit{\small{}Solve scheduling, auxiliary variables and placement with given
$\left\{ t_{i}(k),\,\pi_{i,j}(k),\,\mathcal{S}_{i}(k)\right\} $}{\small \par}

%$\quad$//1)Solve scheduling with given \textit{\small{}$\left\{ \mathcal{S}_{i}(k)\right\} $}{\small \par}

$\quad$\textbf{Step 1}: $\boldsymbol{t}(k+1)=  \underset{\boldsymbol{t}}{\text{\text{argmin}}} (\ref{objeq})$ s.t. \eqref{formf},\eqref{mjt},\eqref{opt_1},\eqref{tjo},\eqref{forml}

$\quad$\textbf{Step 2}: $\boldsymbol{\pi}(k+1)= \underset{\boldsymbol{\pi}}{\text{\text{argmin}}} (\ref{objeq} )$
s.t. \eqref{formf},\eqref{mjt},\eqref{opt_1},\eqref{opt_3}, \eqref{rhocon},\eqref{opt_5},\eqref{forml}

$\quad$//Solve placement with given \textit{\small{}$\left\{\boldsymbol{t}(k+1),\,\boldsymbol{\pi}(k+1)\right\} $}{\small \par}

%$\quad\quad$//Random permutation for files $(1,\ldots,r)$

$\quad$\textbf{Step 3}: 

$ \quad\quad \kappa =\text{random permuation for files \ensuremath{\left(1,\dots,r\right)}}$

$\quad\quad$ \textbf{for} $y = 1, \cdots, r$

$\quad\quad\quad i = \kappa(y)$

$\quad\quad\quad$Calculate $\varLambda_{j}^{\left(-i\right)}$ using $\pi_{ij}(k+1)$

$\quad\quad\quad$Calculate $D_{u,v}$ from (39). 

$\quad\quad\quad$$\left(\beta(u)\;\forall j\right)=\mbox{\ensuremath{Hungarian\,Algorithm(D_{u,v})}}$

$\quad\quad\quad$Update $\pi_{i,\beta(u)}\left(k+1\right)=\pi_{i,u}(k+1)$
$\forall i,j$.

$\quad\quad\quad$ $\mathcal{S}_{i}\left(k+1\right)= \{\beta(u) \forall u \in \mathcal{S}_{i}\left(k\right) \}$

$\quad$\textbf{end for}

\textbf{end while}

\textbf{output: $\left\{ \boldsymbol{\pi},\,\mathcal{\boldsymbol{S}},\,\boldsymbol{t}\right\} $}
\end{algorithm}

%Both these problems are convex, and thus can be solved by Projected Gradient Descent Algorithm. Using these optimizations as the building boxes, the proposed algorithm can be written as follows.

%\begin{enumerate}
%\item {\bf Initialization: } Initialize $\pi_{ij} $ and $t_j$ $\forall$ $(i,j)$ such that the choice is feasible for the problem.
%\item {\bf While Objective Converges: }
%\begin{itemize}
%\item Run {\bf ${\mathbf t}$-Optimization} using current values of ${\bm \pi}$ to get new values of ${\mathbf t}$
%\item Run {\bf ${\bm \pi}$-Optimization} using current values of ${\mathbf t}$ to get new values of ${\bm \pi}$
%\end{itemize}
%\end{enumerate}

Since the objective is non-negative and the objective is non-increasing in each iteration, the algorithm converges.

\section{Numerical Results}

To validate the proposed tail latency bound and tail latency optimization, we employ a hybrid simulation method, which generates chunk service times based on real system measurements on Tahoe and Amazon S3 servers in \cite{Yu_TON,CS14,Yu-TON16}, 
and compare the performance of our proposed latency optimization, denoted as  WLTP Policy, with five baseline strategies. The proposed strategy and the other baseline strategies are described below.

\begin{itemize}[leftmargin=*]
\item Proposed Approach-Optimized Placement, i.e., WLTP ({\em Weighted  Latency Tail Probability}) Policy: The joint scheduler is determined by the proposed solution that minimizes the weighted  latency tail probabilities, with respect to the three sets of variables: chunk placement on the servers ${\bm{  \mathcal{S} }}$, auxiliary variables ${\mathbf t}$, and the scheduling
policy $\bm \pi$.

%our proposed tail latency bounds, the request access optimization and placement optimization. 

\item Proposed Approach-Random Placement, i.e., WLTP-RP ({\em WLTP - Random Placement}) Policy: The chunks are placed uniformly at random. With this fixed placement, the weighted  latency tail probability is optimized over the remaining two sets of variables: auxiliary variables ${\mathbf t}$, and the scheduling
policy $\bm \pi$.

%The joint scheduler is determined by the proposed solution that minimizes the weighted  latency tail probabilities, with respect to our proposed tail latency bounds. 

\item  WLTP-RP-Fixed ${\mathbf t}$ Policy:  The chunks are placed uniformly at random, and all the auxiliary variables $t_j$ are set as $0.01$. The   weighted  latency tail probability is optimized over the scheduling access probabilities  $\bm \pi$.

%The joint request scheduler 
%optimizes the access probabilities while maintaining fixed t, i.e., $t_j$=0.01 $\forall j$. The chunks are placed uniformly at random. 

\item PEAP ({\em Projected, Equal Access-Probability}) Policy: For each file request, the joint request scheduler selects available nodes with equal probability. This choice of $\pi_{i,j}= k_i/n_i$ may not be feasible and thus the  access probabilities are projected toward feasible region in (\ref{objeq}) for all $t_j = .01$ for a uniformly randomly placed files to ensure stability of the storage system. With these fixed access probabilities,  the weighted  latency tail probability is optimized over the remaining two sets of variables: chunk placement on the servers ${\bm{  \mathcal{S} }}$, and the auxiliary variables ${\mathbf t}$.

% and the objective is then optimized over ${\mathbf t}$. we further perform placement optimization besides optimizing auxiliary variables $\boldsymbol{t}$.

\item PEAP-RP Policy: As compared to the PEAP Policy, the chunks are placed uniformly at random. The weighted  latency tail probability is optimized over  the choice of auxiliary variables ${\mathbf t}$.
%This strategy is the same as  PEAP Policy except that no placement optimization is considered. The chunks are placed uniformly at random. 

\item  PSPP ({\em Projected Service-Rate Proportional Allocation}) Policy: The joint request scheduler 
chooses the access probabilities to be proportional to the service rates of the storage nodes, i.e., $\pi_{ij}=k_{i}\frac{\mu_{j}}{\sum_{j}\mu_{j}}$. This policy assigns servers proportional to their service rates. These access probabilities are projected toward feasible region in (\ref{objeq}) for a uniformly random placed files to ensure stability of the storage system.  With these fixed access probabilities,  the weighted  latency tail probability is optimized over the remaining two sets of variables: chunk placement on the servers ${\bm{  \mathcal{S} }}$, and the auxiliary variables ${\mathbf t}$.

%and the objective is then optimized over $\mathbf{t}$ and placement $\boldsymbol{\mathcal{S}}$. 

%is optimized to balance the workload (arrival rates) of all storage nodes as will be described later. Intuitively, this policy minimizes the chance of congested bottleneck in the storage system. The objective is further optimized over ${\mathbf t}$. The strategy here utilizes all $\boldsymbol{\pi}$, $\boldsymbol{t}$ and $\boldsymbol{\mathcal{S}}$ as mentioned above. Then, alternating optimization over placement and 
%$\boldsymbol{t}$ are performed using the
%proposed policy.

\item  PSPP-RP   ({\em  PSPP - Random Placement}) Policy: As compared to the PSPP Policy, the chunks are placed uniformly at random. The weighted  latency tail probability is optimized over  the choice of auxiliary variables ${\mathbf t}$.

%This policy is the same as PSPP except that no placement optimization is performed where the chunks are placed uniformly at random.

\end{itemize}

In the simulations, we consider $r=1000$ files, all of size 200 MB and using $(7,4)$ erasure code in a distributed storage system consisting of $m=12$ distributed nodes. Based on \cite{Yu_TON,CS14,Yu-TON16}, we consider chunk service time that follows a shifted-exponential distribution with rate $\alpha_j$ and shift $\beta_j$. As shown in Table~\ref{tb}, we have 12 heterogeneous storage nodes with different service rates $\alpha_j$ and shifts $\beta_j$. The base arrival rates for the first 250 files is $2/150$ s$^{-1}$, for the next 250 files are $4/150$ s$^{-1}$,  for the next 250 files are $6/150$ s$^{-1}$, and for the last 250 files is $3/150$ s$^{-1}$. %The first 250 files are placed on first seven nodes, the next 250 files are placed on nodes 2 to 8, the next 250 files are placed on nodes 4 to 10, and the last 250 files are placed on nodes 6 to 12.
 This paper also considers the weights of the files proportional to the arrival rates. %and are thu - where the weights corresponding to the first 250 files are each chosen as $2/(15\times 250)$, the weights corresponding to the next 250 files are chosen as $4/(15\times 250)$, the weights corresponding to the next 250 files are chosen as $6/(15\times 250)$, and the weights corresponding to the last 250 files are chosen as $3/(15\times 250)$ such that the sum of weights of all files is $1$. 

 In order to initialize the algorithm, we choose $\pi_{ij} = k/n$ on the placed servers, all $t_j = .01$. %choose a placement that choses the servers uniformly at random, choose the access probabilities as  $\pi_{ij} = k_i/n_i$ on the placed servers, and  all $t_j = .01$. 
 However, since these choices of $\mathbf \pi$ and ${\mathbf t}$ may not be feasible, we modify the initialization ${\bm \pi}$ to be the closest norm feasible solution to the above choice.

\begin{table}[ht]
\caption{Summary of parameters for the 12 storage nodes in our simulation (shift $\beta$ in ms and rate $\alpha$ in 1/s).}\centering
  \resizebox{.48 \textwidth}{!}{
\begin{tabular}{c c c c c c c}
\hline\hline
 & Node 1 & Node 2 & Node 3  & Node 4 & Node 5 & Node 6 \\ [0.25ex]
\hline
$\alpha_j$ & 20.0015  & 26.1252 &  14.9850 & 17.0526  & 27.1422  & 22.8919 \\
$\beta_j$ & 10.5368  & 15.6018 &  8.2756  &  10.0120 &  12.8544 & 13.6722 \\ [1ex]
\hline\hline
 & Node 7 & Node 8 & Node 9  & Node 10 & Node 11 & Node 12 \\ [0.25ex]
\hline
$\alpha_j$ & 30.0000 &  21.3812  &  11.9106 & 25.1599  & 28.8188 & 23.8067  \\
$\beta_j$ & 12.6616   & 9.9156  & 10.7872 & 8.6166  & 13.8721 &  10.8964 \\ [0.25ex]
\hline
\end{tabular}}
\label{tb}
\end{table}

%The $\beta_j$ for different nodes are chosen as 8.5368 ms,   13.6018 ms,    6.2756  ms,  9.5100 ms,   9.0524 ms,  12.1242 ms,  12.3616   ms, 7.4950  ms,  9.9182 ms, 9.5646  ms, 11.1706 ms, and  11.6750 ms respectively. The $\alpha_j$ for different nodes are chosen as 18.2295  s$^{-1}$, 24.0552 s$^{-1}$,   11.8750 s$^{-1}$,  17.0526  s$^{-1}$, 26.1912  s$^{-1}$, 23.9059 s$^{-1}$,  27.0006 s$^{-1}$,  21.3812  s$^{-1}$,  9.9106 s$^{-1}$, 24.9589  s$^{-1}$, 26.5288 s$^{-1}$, and  21.8067 s$^{-1}$ respectively.

%The base arrival rates for the first 500 files are chosen as $0.02$ s$^{-1}$, for the next 5000 files is chosen as $0.03$ s$^{-1}$. The first 250 files are placed on first seven nodes, the next 250 files are placed on nodes 2 to 8, the next 250 files are placed on nodes 4 to 10, and the last 250 files are placed on nodes 6 to 12. This paper also considers different weights of the files - where the weights corresponding to the first 250 files are each chosen as $2/(15\times 250)$, the weights corresponding to the next 250 files are chosen as $4/(15\times 250)$, the weights corresponding to the next 250 files are chosen as $6/(15\times 250)$, and the weights corresponding to the last 250 files are chosen as $3/(15\times 250)$ such that the sum of weights of all files is $1$.  In order to initialize the algorithm, we choose $\pi_{ij} = k/n$ on the placed servers, all $t_j = .01$. However, since these choices of $\mathbf \pi$ and ${\mathbf t}$ may not be feasible, we modify the initialization ${\bm \pi}$ to be the closest norm feasible solution to the above choice.

\begin{figure}
	\begin{center}
		\includegraphics[trim=.0in 0in 5.7in .0in, clip,width=0.45\textwidth]{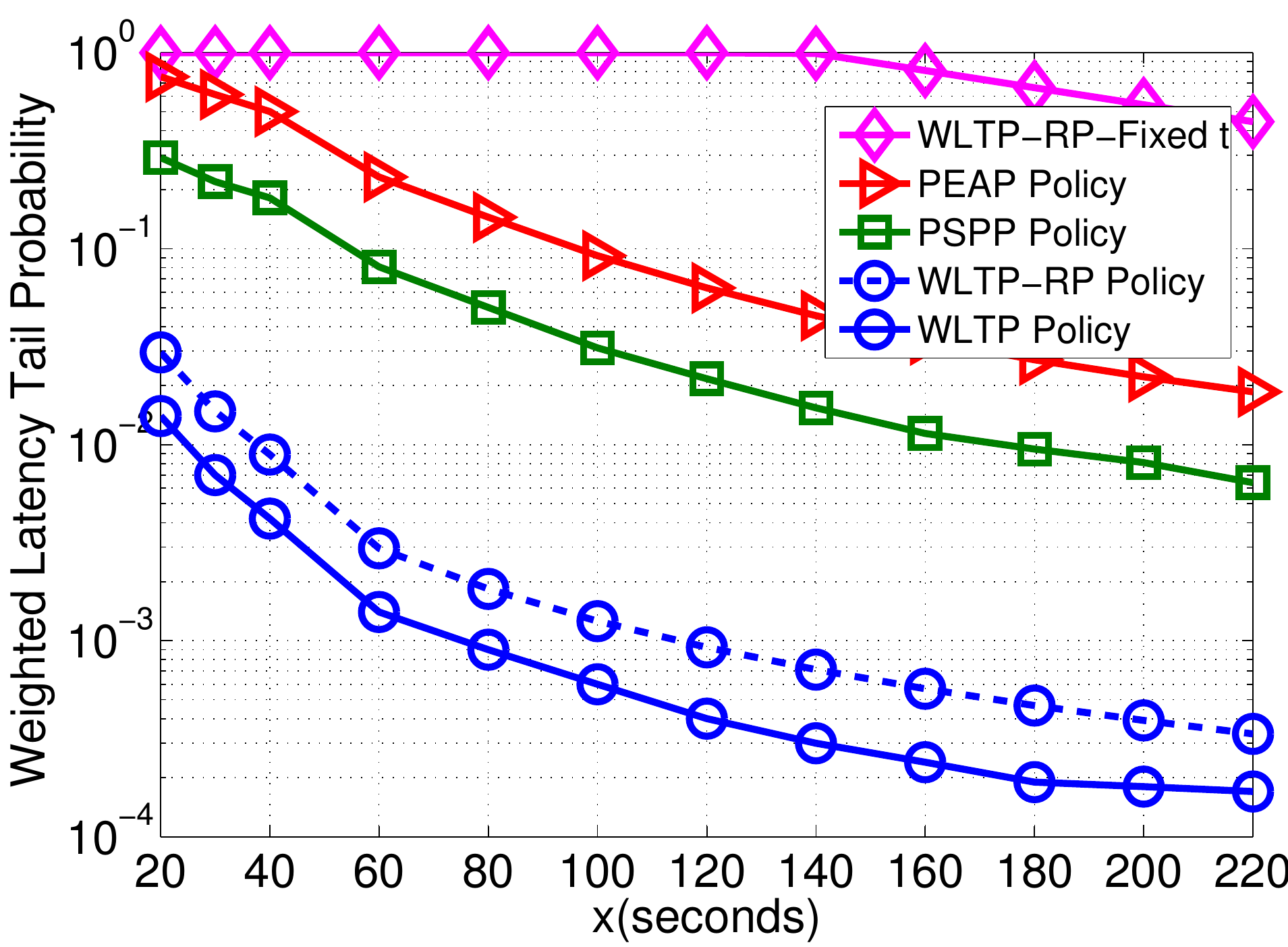}
		\caption{Weighted  Latency Tail Probability vs $x$ (in seconds).}
		\label{tailfig}
	\end{center}
\end{figure}

\noindent {\bf Weighted  Latency Tail Probabilities: } In Figure~\ref{tailfig}, we plot the decay of weighted  latency tail probability $\sum_i \omega_i\Pr(L_i \ge x)$ with $x$ (in seconds) for Policies WLTP, WLTP-RP, PSPP, PEAP, and WLTP-RP-Fixed ${\mathbf t}$. Notice that WLTP Policy solves the optimal weighted latency tail probability via proposed alternative optimization algorithm over  $\bm \pi$, ${\mathbf t}$, and  ${\bm{  \mathcal{S} }}$. With optimized ${\mathbf t}$ and placement, Policy PEAP uses equal server access probabilities, projected toward the feasible region, while Policy PSPP assign chunk requests across different servers proportional to their service rates. The values of ${\mathbf t}$ are then found optimally for the above given values of $\pi_{i,j}$. Note that this figure also represents the complementary cumulative distribution function (ccdf) of the WLTP, WLTP-RP, WLTP-RP-Fixed ${\mathbf t}$, PSPP, and PEAP. For instance, We observe that $\text{Pr}\left(x\geq20\right)\approx0.01$ for our proposed policy WLTP which is significantly lower as compared to the other strategies.

%Figure  \ref{CDF} shows the  cumulative distribution function (CDF) of the WLTP, WLTP-RP, PSPP, PEAP and Optimized Access-Fixed t policies. This figure shows that   
%the WLTP policy achieves the lowest probability and Optimized Access-fixed t policy performs the worst. We note that Pr$(x\leq24)=1$ for our proposed policy, i.e., WLTP.

%we have the first 250 files access the first 4 servers with equal probabilities (=1), the last 250 files access the last 4 servers with equal probabilities (=1), whereas files 251 to 500 access the 12 servers with probabilities [0 17 17 17 18 3 0 0 0 0 0 0]/18, and files 501 to 750 access the servers with probabilities [0 0 0 0 16  46 52 52 25 25 0 0 ]/54. With this choice, the aggregate arrival rate at the first server is 3.33, at the last two servers is 5 and the rest 9 servers is 9.6296. 

%servers arrival rate at the first, 11th and 12th server can no longer be increased as each hosts only a single file.

%As a comparison, we have the initial choice of the access probabilities ${\bm \pi}$, and optimize over $t$. The other comparison is where we fix the ${\mathbf t}$ as in the initialization and optimize over the access probabilities ${\bm \pi}$.

We note that our proposed algorithm for jointly optimizing ${\bm \pi}$, ${\mathbf t}$ and $\boldsymbol{\mathcal{S}}$ provides significant improvement over simple heuristics such as Policies WLTP-RP-Fixed ${\mathbf t}$, PSPP, and PEAP, as weighted  latency tail probability reduces by  orders of magnitude. For example, our proposed Policy WLTP decreases 99-percentile weighted latency (i.e., $x$ such that $\sum_i \omega_i\Pr(L_i \ge x)\le 0.01$) from above 160 seconds in the baseline policies to about 20 seconds using WLTP.  We also notice that chunk placement optimization reduces  the latency tail probability  for all $x$, as  can be seen from Figure~\ref{tailfig} among the policies WLTP and WLTP-RP. Uniformly accessing servers and simple service-rate-based scheduling are unable to optimize the request scheduler based on factors like chunk placement, request arrival rates, different latency weights, thus leading to much higher tail latency. Since the policy WLTP-RP-Fixed ${\mathbf t}$ performs significantly worse than the other considered policies, we do not include this policy in the rest of the paper.

\begin{figure}
	\begin{center}
		\includegraphics[trim=.1in .1in 6.1in .1in, clip,width=0.45\textwidth]{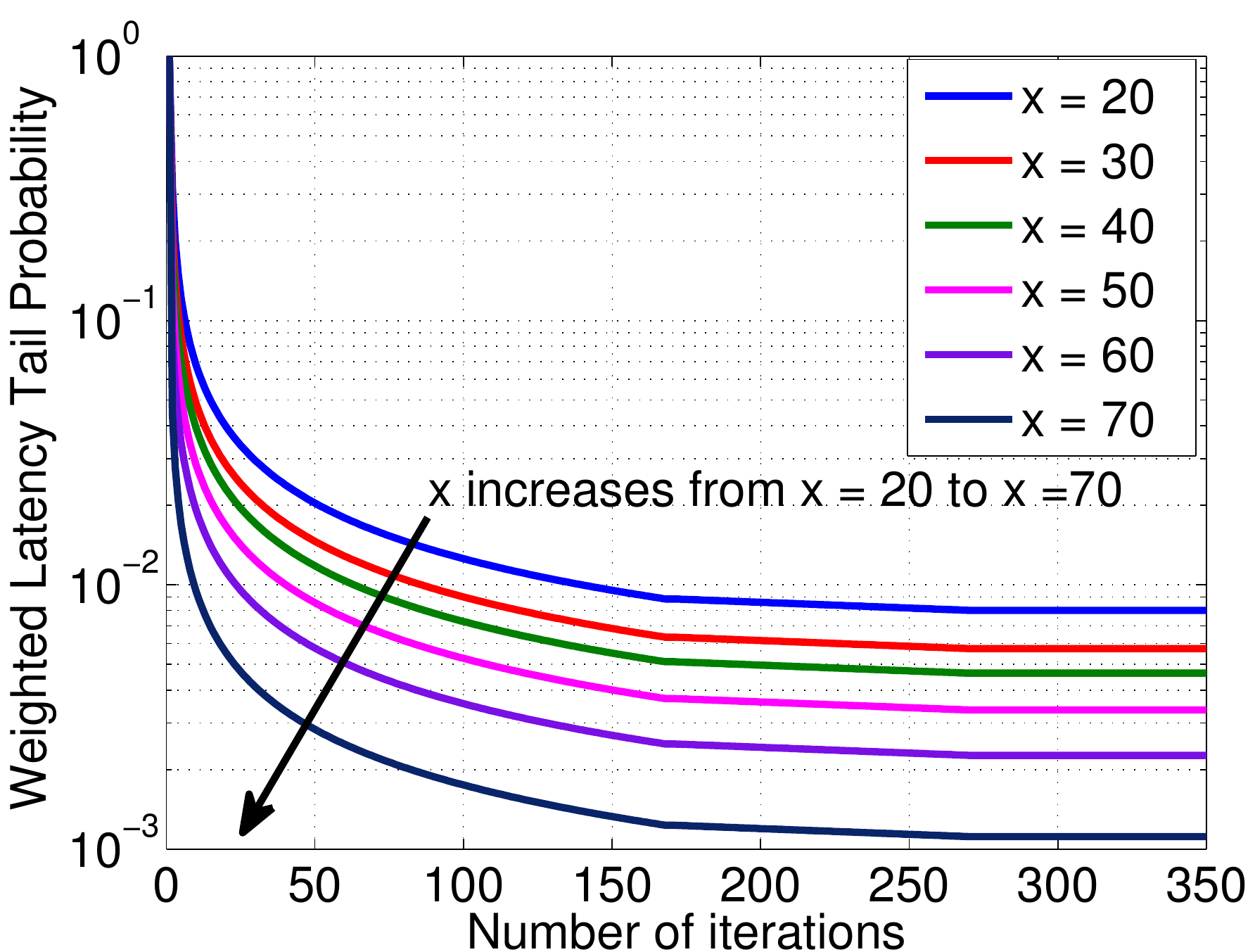}
		\caption{Convergence of Weighted Latency Tail Probability.}
		\label{convfig}
	\end{center}
\end{figure}

\noindent {\bf Convergence of the Proposed Algorithm: } We have shown that the proposed algorithm converges within about 350 iterations to the optimal objective value, validating the efficiency of the proposed optimization algorithm. To illustrate its convergence speed, Figure \ref{convfig} shows the convergence of objective value vs. the number of iterations for different values of $x$ ranging from 20 to 70 seconds in increments of 10 seconds. %For 1000 files and 12 storage nodes, we note that the weighted  latency tail probability shows convergence. 
In the rest of the results, 350 iterations will be used to get the required results.

%and the results can only be better with higher number of iterations. 

%convergences very fast and within about 300 iterations to the optimal objective value, validating the efficiency of the proposed optimization algorithm.

%\begin{figure}
%	\begin{center}
%		\includegraphics[trim=0in .0in 0in 0in, clip,width=0.80\textwidth]{CDF}
%		\caption{CDF for the different policies, i.e., WLTP, WLTP-RP, PEAP and PSPP,  }
%		\label{CDF}
%	\end{center}
%\end{figure}

%Figure  \ref{CDF} shows the  cumulative distribution function (CDF) of the WLTP, WLTP-RP, PSPP, PEAP and Optimized Access-Fixed t policies. This figure shows that   
%the WLTP policy achieves the lowest probability and Optimized Access-fixed t policy performs the worst. We note that Pr$(x\leq24)=1$ for our proposed policy, i.e., WLTP.

\begin{figure}
	\begin{center}
		\includegraphics[trim=.0in .0in 5.3in .0in, clip,width=0.45\textwidth]{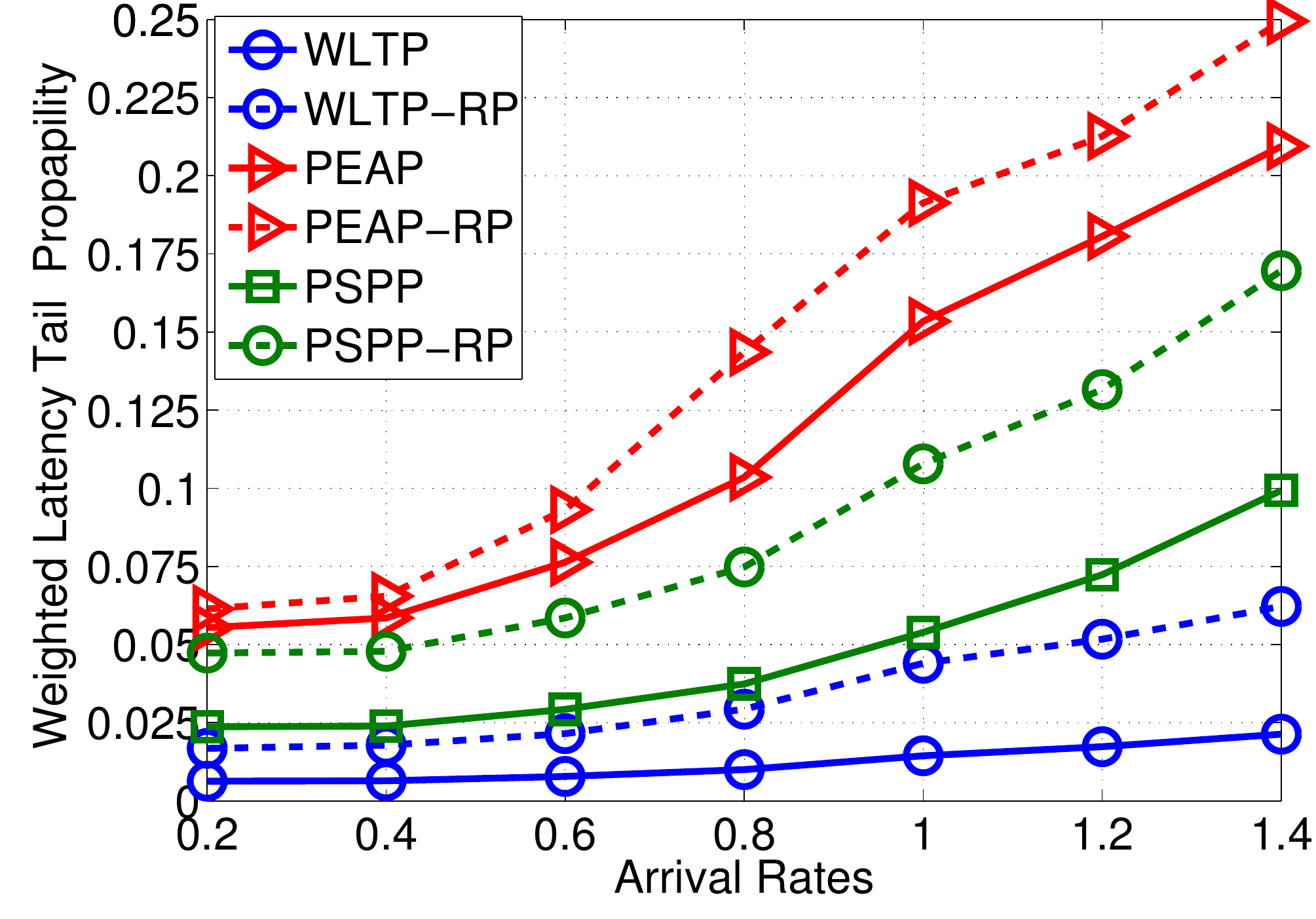}
		\caption{Weighted Latency Tail Probability for different file arrival rates. We vary
the arrival rate of file i from $0.2\times\lambda$ to $1.4\times\lambda$, where $\lambda$ is the base
arrival rate.}
		\label{arrfig}
	\end{center}
\end{figure}

\noindent {\bf Effect of Arrival Rates: } We next see the effect of varying request arrival rates on the weighted latency tail probability. We choose $x=50$ seconds. For $\lambda$ as the base arrival rates, we increase arrival rate of all files from $.2\lambda$ to $1.4\times\lambda$ and plot the weighted  latency tail probability in Figure~\ref{arrfig}. While latency tail probability increases as arrival rate increases, our algorithm assigns differentiated latency for different files to maintain low weighted latency tail probability. We compare
our proposed algorithm with the different baseline policies and notice that the proposed algorithm outperforms all baseline strategies.

Since the weighted latency tail probability is more significant at high arrival
rates, we observe significant improvement in latency tail
by about a multiple of 9 ( approximately 0.025 to about 0.22) at the highest
arrival rate in Figure~\ref{arrfig} between PEAP and WLTP policies. Our proposed policy always receives the minimum latency. Thus, efficiently reducing the latency of the high arrival rate files reduces the overall weighted latency tail probability.  

%We consider the arrival rates from $.6\lambda$ to $\lambda$ in Figure~\ref{arrfig}.  We see that the weighted tail latency increases with the arrival rate. Further, the third set of 250 files have a much higher weight and thus the algorithm tries to decrease the latencies of this set of files significantly at the expense of the other files.

\begin{figure}
	\begin{center}
		\includegraphics[trim=.0in .1in 5.8in .0in,clip,width=0.45\textwidth]{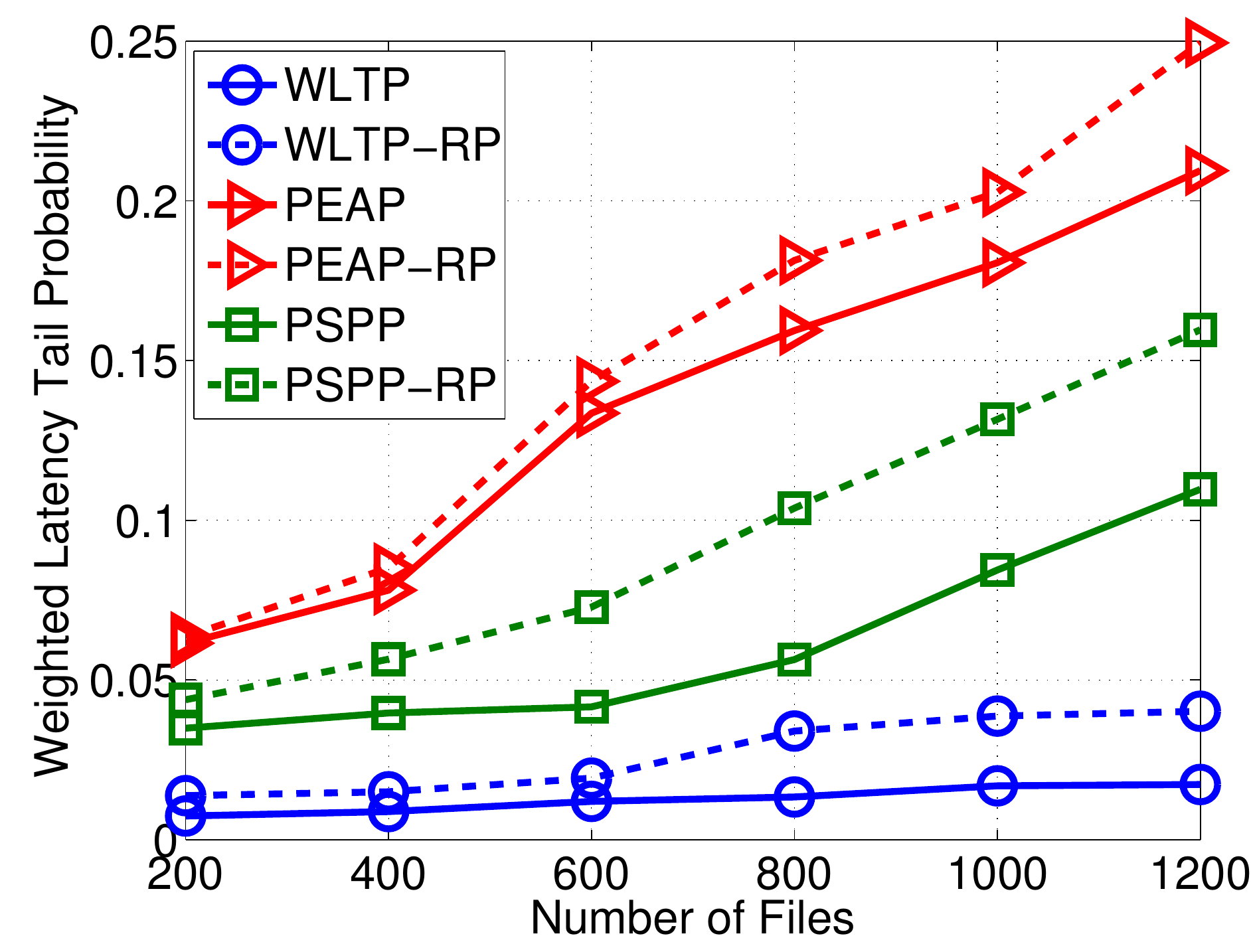}
		\caption{Weighted Latency Tail Probability for different number of files.}
		\label{filesfig}
	\end{center}
\end{figure}

\noindent {\bf Effect of Number of files: } Figure \ref{filesfig} demonstrates the impact of varying the number of  files from 200  to 1200  on the weighted latency tail probability. Weighted latency tail  probabilities increases with the number of files, which brings in more workload (i.e., higher arrival rates). Our optimization algorithm optimizes new files along with existing ones to keep overall weighted latency tail probability at a  low level. We note that the proposed optimization strategy effectively reduces the tail probability and outperforms the considered baseline strategies. Thus, joint optimization over all three variables $\boldsymbol{\mathcal{S}}$, $\boldsymbol{\pi}$, and $\boldsymbol{t}$ helps reduce the tail probability significantly.

\begin{figure}
	\begin{center}
		\includegraphics[trim=0in .1in 5.8in .0in, clip,width=0.45\textwidth]{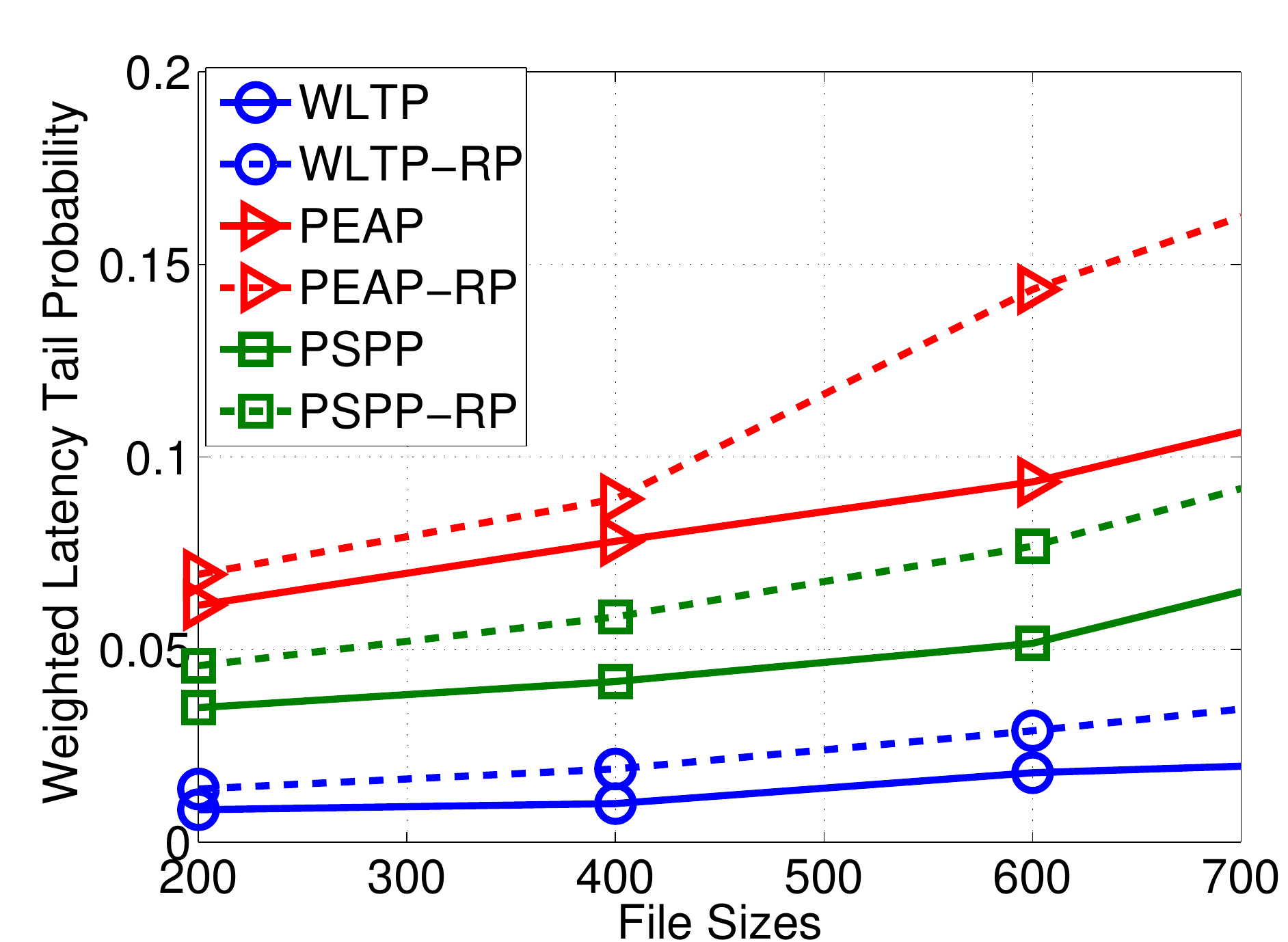}
		\caption{Weighted Latency Tail Probability for different file sizes.}
		\label{filesizefig}
	\end{center}
\end{figure}

\noindent {\bf Effect of File Sizes: } We vary the file size in our simulation from $200$MB to $700$MB, and plot the optimal weighted latency tail probability with varying file size in Figure \ref{filesizefig}. In order to capture the effect of increased file size as compared to a default size of 200 MB, we increase the value of parameters $\alpha$ and $\beta$ proportionally to the chunk size in the shifted-exponential service time distribution. While increasing file size results in higher weighted tail latency probability for files, we compare our proposed algorithm with the baseline policies and verified that the proposed optimization algorithm offers significant reduction in tail latency.
\begin{figure}
	\begin{center}
		\includegraphics[trim=1.3in .1in 1.2in .0in, clip,width=0.45\textwidth]{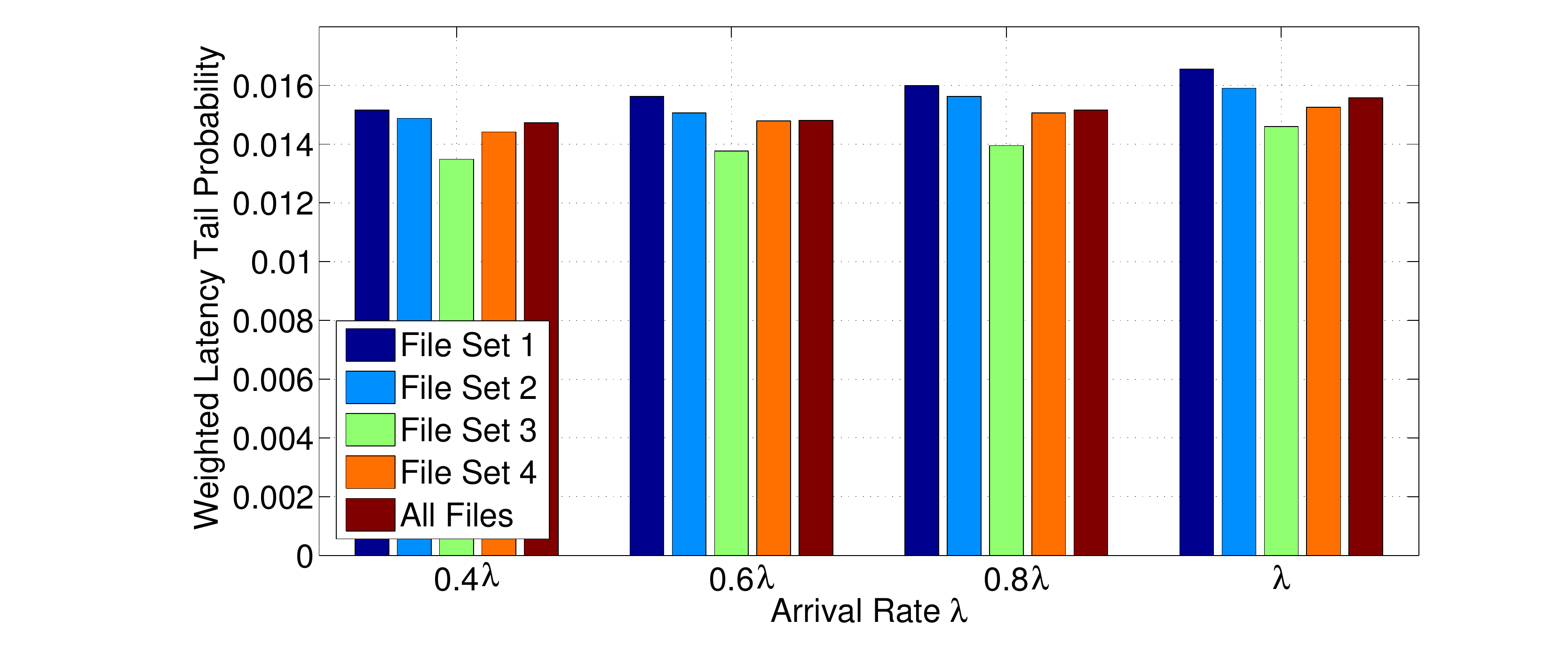}
		\caption{Weighted Latency Tail Probability for different file arrival rates.}
		\label{arrfig1}
	\end{center}
\end{figure}

\noindent {\bf Effect of the Tail Latency Weights}:  We next show the effect of varying the weights (i.e., $w_i$'s) on the weighted tail latency  probability for the proposed WLTP policy. We choose $x=40$ seconds. 
%
% and divide all files into 4 groups, each containing 250 consecutive files of equal weight. We set $\lambda$=0.02 for the first group (the first 250 files), $\lambda$=0.03 for the next 250 files group, $\lambda$=0.05 for the third 250 files group, and finally, $\lambda$=0.04 for the last 250 files group.  
Recall that the arrival rate of files was divided into four groups, each with different arrival rates. We vary the arrival rate of all files from $.4\lambda$ to $\lambda$ with a step of $0.2\lambda$ and plot the weighted  latency tail probability for each group of 250 files as well as the overall value in Figure~\ref{arrfig1}. While weighted latency tail probability increases as arrival rate increases, our algorithm assigns differentiated latency for different file groups. Group 3 that has highest weight $\omega_3$ (i.e., most tail latency sensitive) always receive the minimum latency tail probability even though these files have the highest arrival rate. Thus, efficiently reducing the latency of the high arrival rate files   reduces the overall weighted latency tail probability.  We note that  efficient access probabilities $\bm \pi$ help in differentiating file latencies as compared to the strategy where minimum queue-length  servers are selected to access the content obtaining lower weighted tail latency probability.

%We consider the arrival rates from $.6\lambda$ to $\lambda$ in Figure~\ref{arrfig}.  We see that the weighted tail latency increases with the arrival rate. Further, the third set of 250 files have a much higher weight and thus the algorithm tries to decrease the latencies of this set of files significantly at the expense of the other files.

\begin{figure}
	\begin{center}
		\includegraphics[trim=0.5in .1in 1.0in .0in, clip,width=0.5\textwidth]{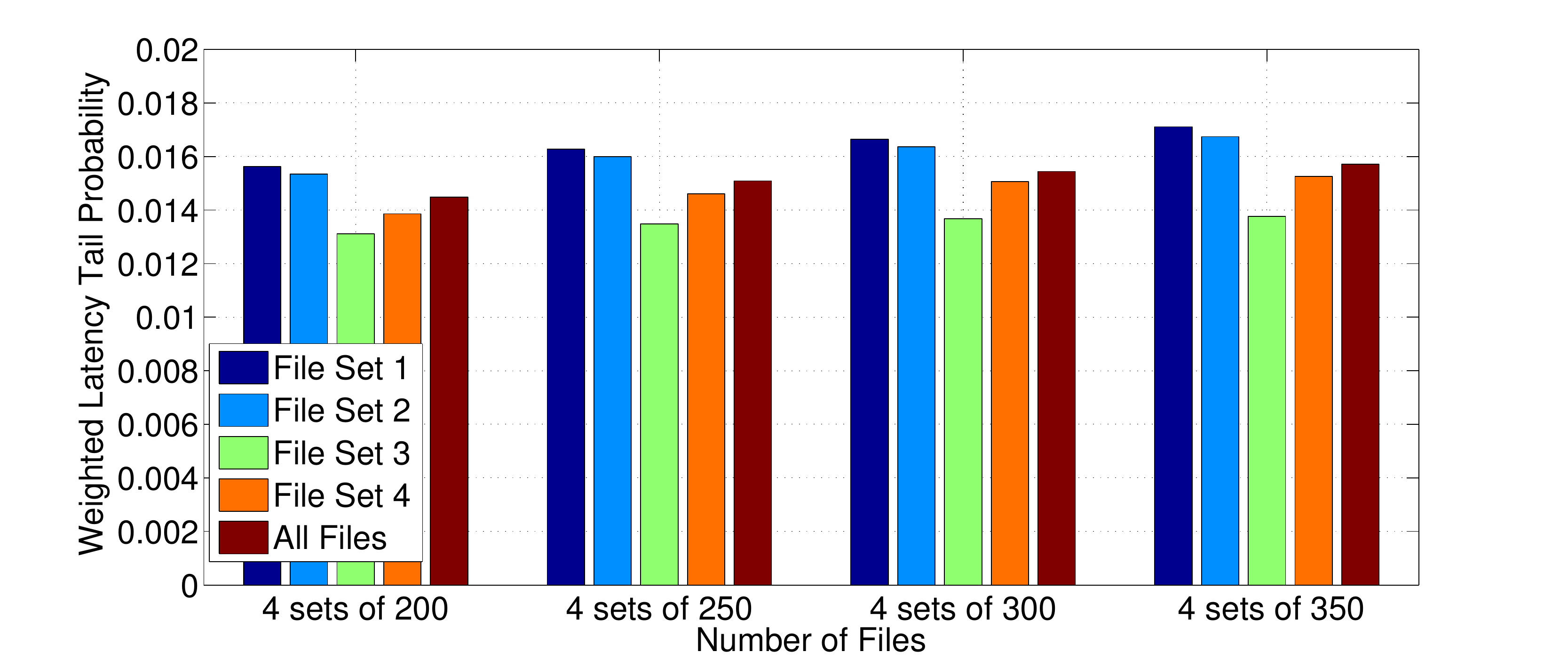}
		\caption{Weighted Latency Tail Probability for different file arrival rates.}
		\label{NoOfFiles}
	\end{center}
\end{figure}

Figure~ \ref{NoOfFiles} shows the effect of the file weights on the weighted latency tail probability for varying number of files. In particular, we modify the number of files in each group from $250$ in the base case to values such
as $250$, $300$, and $350$, as shown in  Figure~ \ref{NoOfFiles}. Apparently, the weighted tail latency  probability increases with the number of files,
since that increases the workload in terms of the arrival rates.
Our optimization algorithm optimizes the placement and access of the files to keep weighted  tail  latency probability lower. As noted earlier, the higher arrival rate group has lower tail latency probability thus reducing the overall objective. 
\section{Acknowledgment}
This work was supported in
part by the National Science Foundation under Grant no. CNS-1618335.

\section{Conclusions}
This paper provides bounds on latency tail probabilities for distributed storage systems using erasure coding. These bounds are used to formulate an optimization to jointly minimize weighted latency tail probability of all files over request scheduling and data placement. The non-convex optimization problem is solved using an efficient alternating optimization algorithm. Simulation results show significant reduction of tail latency for erasure-coded storage systems with realistic workload.

Following this work, the probabilistic scheduling used in this paper has been shown to be optimal for tail index in \cite{Pareto}. However, the model of file size is different as compared to this paper. Finding more general scenarios where such scheduling strategy is optimal, or improving the strategy to show optimality for wider classes is an important research problem.
\bibliographystyle{IEEEtran}
%\bibliography{Vaneet_cloud,allstorage, Tian_rest}

\bibliography{allstorage,Tian,ref_Tian2,ref_Tian3,Vaneet_cloud,Tian_rest}

%\appendices
%\input{beck}
\end{document}